\documentclass[fleqn,12pt,twoside]{article}

\usepackage[headings]{espcrc1}
\readRCS
$Id: espcrc1.tex,v 1.2 2004/02/24 11:22:11 spepping Exp $
\usepackage{graphicx}
\usepackage[figuresright]{rotating}

\newtheorem{theorem}{Theorem}

\newtheorem{corollary}[theorem]{Corollary}

\newtheorem{lemma}[theorem]{Lemma}

\newtheorem{proposition}[theorem]{Proposition}
\newtheorem{remark}[theorem]{Remark}

\newenvironment{proof}[1][Proof.]{\begin{trivlist}
\item[\hskip \labelsep {\bfseries #1}]}{\end{trivlist}}

\newenvironment{acknowledgement}[1][Acknowledgement]{\begin{trivlist}
\item[\hskip \labelsep {\bfseries #1}]}{\end{trivlist}}

\newcommand{\AmS}{{\protect\the\textfont2
  A\kern-.1667em\lower.5ex\hbox{M}\kern-.125emS}}

\usepackage{amsmath}
\usepackage{amsfonts}
\title{Interval total colorings of graphs}

\author{P.A. Petrosyan\address[MCSD]{Department of Informatics and Applied Mathematics,\\
Yerevan State University, 0025, Armenia}%
\address{Institute for Informatics and Automation Problems,\\
National Academy of Sciences, 0014, Armenia}%
\thanks{email: pet\_petros@ipia.sci.am},
        A.Yu. Torosyan\addressmark[MCSD]%
\thanks {email: arman.yu.torosyan@gmail.com},
        N.A. Khachatryan\addressmark[MCSD]%
\thanks {email: xachnerses@gmail.com}}

\runtitle{Interval total colorings of graphs}\runauthor{P.A.
Petrosyan, A.Yu. Torosyan, N.A. Khachatryan}

\begin{document}

\maketitle

\begin{abstract}
A total coloring of a graph $G$ is a coloring of its vertices and
edges such that no adjacent vertices, edges, and no incident
vertices and edges obtain the same color. An \emph{interval total
$t$-coloring} of a graph $G$ is a total coloring of $G$ with colors
$1,2,\ldots,t$ such that at least one vertex or edge of $G$ is
colored by $i$, $i=1,2,\ldots,t$, and the edges incident to each
vertex $v$ together with $v$ are colored by $d_{G}(v)+1$ consecutive
colors, where $d_{G}(v)$ is the degree of the vertex $v$ in $G$. In
this paper we investigate some properties of interval total
colorings. We also determine exact values of the least and the
greatest possible number of colors in such colorings for some
classes of graphs.\\


\end{abstract}

\section{Introduction}\

A total coloring of a graph $G$ is a coloring of its vertices and
edges such that no adjacent vertices, edges, and no incident
vertices and edges obtain the same color. The concept of total
coloring was introduced by V. Vizing \cite{b22} and independently by
M. Behzad \cite{b4}. The total chromatic number $\chi^{\prime
\prime}\left(G\right)$ is the smallest number of colors needed for
total coloring of $G$. In 1965 V. Vizing and M. Behzad conjectured
that $\chi^{\prime\prime}\left( G\right)\leq \Delta (G)+2$ for every
graph $G$ \cite{b4,b22}, where $\Delta (G)$ is the maximum degree of
a vertex in $G$. This conjecture became known as Total Coloring
Conjecture \cite{b10}. It is known that Total Coloring Conjecture
holds for cycles, for complete graphs \cite{b5}, for bipartite
graphs, for complete multipartite graphs \cite{b25}, for graphs with
a small maximum degree \cite{b11,b12,b18,b21}, for graphs with
minimum degree at least $\frac{3}{4}\vert V(G)\vert$ \cite{b9}, and
for planar graphs $G$ with $\Delta(G)\neq 6$ \cite{b6,b10,b20}. M.
Rosenfeld \cite{b18} and N. Vijayaditya \cite{b21} independently
proved that the total chromatic number of graphs $G$ with $\Delta
(G)=3$ is at most $5$. A. Kostochka in \cite{b11} proved that the
total chromatic number of graphs with $\Delta(G)=4$ is at most $6$.
Later, also he in \cite{b12} proved that the total chromatic number
of graphs with $\Delta(G)=5$ is at most $7$. The general upper bound
for the total chromatic number was obtained by M. Molloy and B. Reed
\cite{b15}, who proved that $\chi^{\prime\prime}\left( G\right)\leq
\Delta (G)+10^{26}$ for every graph $G$. The exact value of the
total chromatic number is known only for paths, cycles, complete and
complete bipartite graphs \cite{b5}, $n$-dimensional cubes, complete
multipartite graphs of odd order \cite{b8}, outerplanar graphs
\cite{b26} and planar graphs $G$ with $\Delta (G)\geq 9$
\cite{b7,b10,b13,b23}.

The key concept discussed in a present paper is the following. Given
a graph $G$, we say that $G$ is interval total colorable if there is
$t\geq 1$ for which $G$ has a total coloring with colors $1,2,\ldots
,t$ such that at least one vertex or edge of $G$ is colored by $i$,
$1,2,\ldots,t$, and the edges incident to each vertex $v$ together
with $v$ are colored by $d_{G}(v)+1$ consecutive colors, where
$d_{G}(v)$ is the degree of the vertex $v$ in $G$.

The concept of interval total coloring \cite{b16,b17} is a new one
in graph coloring, synthesizing interval colorings \cite{b1,b2} and
total colorings. The introduced concept is valuable as it connects
to the problems of constructing a timetable without a
\textquotedblleft gap\textquotedblright and it extends to total
colorings of graphs one of the most important notions of classical
mathematics - the one of continuity.

In this paper we investigate some properties of interval total
colorings of graphs. Also, we show that simple cycles, complete
graphs, wheels, trees, regular bipartite graphs and complete
bipartite graphs have interval total colorings. Moreover, we obtain
some bounds for the least and the greatest possible number of colors
in interval total colorings of these graphs.
\bigskip

\section{Definitions and preliminary results}\

All graphs considered in this work are finite, undirected, and have
no loops or multiple edges. Let $V(G)$ and $E(G)$ denote the sets of
vertices and edges of $G$, respectively. An $(a,b)$-biregular
bipartite graph $G$ is a bipartite graph $G$ with the vertices in
one part having degree $a$ and the vertices in the other part having
degree $b$. The degree of a vertex $v\in V(G)$ is denoted by
$d_{G}(v)$, the maximum degree of vertices in $G$ by $\Delta (G)$,
the diameter of $G$ by $diam(G)$, the chromatic number of $G$ by
$\chi(G)$ and the edge-chromatic number of $G$ by
$\chi^{\prime}(G)$. A vertex $u$ of a graph $G$ is universal if
$d_{G}(u)=\vert V(G)\vert -1$. A \emph{proper edge-coloring} of a
graph $G$ is a coloring of the edges of $G$ such that no two
adjacent edges receive the same color. For a total coloring $\alpha$
of a graph $G$ and for any $v\in V(G)$, define the set
$S\left[v,\alpha\right]$ as follows:

\begin{center}
$S\left[v,\alpha \right]= \left\{\alpha (v)\right\}\cup \left\{
\alpha (e)\left\vert \text{ }e\text{ is incident to }v\right.
\right\}$
\end{center}

Let $\left\lfloor a\right\rfloor $ ($\left\lceil a\right\rceil $)\
denote the greatest (the least) integer $\leq a$ ($\geq a$). For two
integers $a\leq b$, the set $\left\{a,a+1,\ldots ,b\right\}$ is
denoted by $\left[a,b\right]$.

An \emph{interval $t$-coloring} of a graph $G$ is a proper
edge-coloring of $G$ with colors $1,2,\ldots,t$ such that at least
one edge of $G$ is colored by $i$, $i=1,2,\ldots,t$, and the edges
incident to each vertex $v$ are colored by $d_{G}(v)$ consecutive
colors. A graph $G$ is interval colorable if there is $t\geq 1$ for
which $G$ has an interval $t$-coloring. The set of all interval
colorable graphs is denoted by $\mathfrak{N}$. For a graph $G\in
\mathfrak{N}$, the greatest value of $t$ for which $G$ has an
interval $t$-coloring is denoted by $W\left(G\right)$.

An \emph{interval total $t$-coloring} of a graph $G$ is a total
coloring of $G$ with colors $1,2,\ldots,t$ such that at least one
vertex or edge of $G$ is colored by $i$, $i=1,2,\ldots,t$, and the
edges incident to each vertex $v$ together with $v$ are colored by
$d_{G}(v)+1$ consecutive colors.

For $t\geq 1$, let $\mathfrak{T}_{t}$ denote the set of graphs which
have an interval total $t$-coloring, and assume: $\mathfrak{T}=
\underset{t\geq 1}{\bigcup} \mathfrak{T}_{t}$. For a graph $G\in
\mathfrak{T}$, the least and the greatest values of $t$ for which
$G\in \mathfrak{T}_{t}$ are denoted by $w_{\tau }\left(G\right)$ and
$W_{\tau }\left( G\right)$, respectively. Clearly,

\begin{center}
$\chi^{\prime \prime}\left(G\right)\leq w_{\tau }\left( G\right)\leq
W_{\tau }\left(G\right)\leq \vert V(G)\vert + \vert E(G)\vert$ for
every graph $G\in \mathfrak{T}$.
\end{center}

Terms and concepts that we do not define can be found in
\cite{b24,b25}.

We will use the following two results.

\begin{theorem}
\label{mytheorem1}\cite{b1,b2}. If $G$ is a connected triangle-free
graph and $G\in \mathfrak{N}$, then
\begin{center}
$W(G)\leq \vert V(G)\vert -1$.
\end{center}
\end{theorem}

\begin{theorem}
\label{mytheorem2}\cite{b3}. If $G$ is a connected $(a,b)$-biregular
bipartite graph with $\vert V(G)\vert \geq 2(a+b)$ and $G\in
\mathfrak{N}$, then
\begin{center}
$W(G)\leq \vert V(G)\vert -3$.
\end{center}
\end{theorem}
\bigskip

\section{Some properties of interval total colorings of graphs}\

First we prove a simple property of interval total colorings that
for any interval total coloring of a graph $G$ there is an inverse
interval total coloring of the same graph.

\begin{proposition}
\label{myproposition1}
If $\alpha$ is an interval total $t$-coloring of a graph $G$, then a
total coloring $\beta$, where

1) $\beta(v)=t+1-\alpha(v)$ for each $v\in V(G)$,

2) $\beta(e)=t+1-\alpha(e)$ for each $e\in E(G)$,\\
is also an interval total $t$-coloring of a graph $G$.
\end{proposition}
\begin{proof} Clearly, a total coloring $\beta$ contains at least
one vertex or edge with color $i$, $i=1,2,\ldots,t$. Since
$S[v,\alpha]$ is an interval for each $v\in V(G)$, hence
$S[v,\alpha]=[a,b]$. By the definition of the coloring $\beta$ it
follows that $S[v,\beta]=[t+1-b,t+1-a]$ for each $v\in V(G)$.
 $~\square$
\end{proof}

Next we prove the proposition which implies that in definition of
interval total $t$-coloring, the requirement that every color $i$,
$i=1,2,\ldots,t$, appear in an interval total $t$-coloring isn't
necessary in the case of connected graphs.

\begin{proposition}
\label{myproposition2} Let $\alpha$ be a total coloring of the
connected graph $G$ with colors $1,2,\ldots,t$ such that the edges
incident to each vertex $v\in V(G)$ together with $v$ are colored by
$d_{G}(v)+1$ consecutive colors, and $\min_{v\in V(G),e\in
E(G)}\{\alpha(v),\alpha(e)\}=1$, $\max_{v\in V(G),e\in
E(G)}\{\alpha(v),\alpha(e)\}=t$. Then $\alpha$ is an interval total
$t$-coloring of $G$.
\end{proposition}
\begin{proof}
For the proof of the proposition it suffices to show that if $t\geq
3$, then for color $s$, $1<s<t$, there exists at least one vertex or
edge of $G$ which is colored by $s$. We consider four possible
cases.

Case 1: there are vertices $v,v^{\prime}\in V(G)$ such that
$\alpha(v)=1$, $\alpha(v^{\prime})=W_{\tau}(G)$.

Since $G$ is connected, there exists a simple path $P_{1}$ joining
$v$ with $v^{\prime}$, where
\begin{center}
$P_{1}=(v_{0},e_{1},v_{1},\ldots,v_{i-1},e_{i},v_{i},\ldots,v_{k-1},e_{k},v_{k})$,
$v_{0}=v$, $v_{k}=v^{\prime}$.
\end{center}

If $\alpha(v_{i})\neq s$, $i=1,2,\ldots,k-1$, and $\alpha(e_{j})\neq
s$, $j=1,2,\ldots,k$, then there exists an index $i_{0}$, $1\leq
i_{0}<k$, such that $\alpha(e_{i_{0}})<s$ and
$\alpha(e_{i_{0}+1})>s$. Hence, there is an edge of $G$ colored by
$s$ which is incident to $v_{i_{0}}$. This implies that for any $s$,
$1<s<t$, there is a vertex or an edge with color $s$.

Case 2: there is a vertex $v$ and there is an edge $e^{\prime}$ such
that $\alpha(v)=1$, $\alpha(e^{\prime})=W_{\tau}(G)$.

Let $e^{\prime}=v^{\prime}w$ and $P_{2}$ be a simple path joining
$v$ with $v^{\prime}$, where
\begin{center}
$P_{2}=(v_{0},e_{1},v_{1},\ldots,v_{i-1},e_{i},v_{i},\ldots,v_{k-1},e_{k},v_{k})$,
$v_{0}=v$, $v_{k}=v^{\prime}$.
\end{center}

If $\alpha(v_{i})\neq s$, $i=1,2,\ldots,k$, and $\alpha(e_{j})\neq
s$, $j=1,2,\ldots,k$, then there exists an index $i_{1}$, $1\leq
i_{1}<k$, such that $\alpha(e_{i_{1}})<s$ and
$\alpha(e_{i_{1}+1})>s$. Hence, there is an edge of $G$ colored by
$s$ which is incident to $v_{i_{1}}$. This implies that for any $s$,
$1<s<t$, there is a vertex or an edge with color $s$.

Case 3: there is an edge $e$ and there is a vertex $v^{\prime}$ such
that $\alpha(e)=1$, $\alpha(v^{\prime})=W_{\tau}(G)$.

Let $e=uv$ and $P_{3}$ be a simple path joining $v$ with
$v^{\prime}$, where
\begin{center}
$P_{3}=(v_{0},e_{1},v_{1},\ldots,v_{i-1},e_{i},v_{i},\ldots,v_{k-1},e_{k},v_{k})$,
$v_{0}=v$, $v_{k}=v^{\prime}$.
\end{center}

If $\alpha(v_{i})\neq s$, $i=0,1,\ldots,k-1$, and $\alpha(e_{j})\neq
s$, $j=1,2,\ldots,k$, then there exists an index $i_{2}$, $1\leq
i_{2}<k$, such that $\alpha(e_{i_{2}})<s$ and
$\alpha(e_{i_{2}+1})>s$. Hence, there is an edge of $G$ colored by
$s$ which is incident to $v_{i_{2}}$. This implies that for any $s$,
$1<s<t$, there is a vertex or an edge with color $s$.

Case 4: there are edges $e,e^{\prime}\in E(G)$ such that
$\alpha(e)=1$, $\alpha(e^{\prime})=W_{\tau}(G)$.

Let $e=uv$, $e=v^{\prime}w$. Without loss of generality we may
assume that a simple path $P_{4}$ joining $e$ with $e^{\prime}$
joins $v$ with $v^{\prime}$, where
\begin{center}
$P_{4}=(v_{0},e_{1},v_{1},\ldots,v_{i-1},e_{i},v_{i},\ldots,v_{k-1},e_{k},v_{k})$,
$v_{0}=v$, $v_{k}=v^{\prime}$.
\end{center}

If $\alpha(v_{i})\neq s$, $i=0,1,\ldots,k$, and $\alpha(e_{j})\neq
s$, $j=1,2,\ldots,k$, then there exists an index $i_{3}$, $1\leq
i_{3}<k$, such that $\alpha(e_{i_{3}})<s$ and
$\alpha(e_{i_{3}+1})>s$. Hence, there is an edge of $G$ colored by
$s$ which is incident to $v_{i_{3}}$. This implies that for any $s$,
$1<s<t$, there is a vertex or an edge with color $s$. $~\square$
\end{proof}

Now we show that there is an intimate connection between interval
total colorings of graphs and interval edge-colorings of certain
bipartite graphs.

Let $G$ be a simple graph with $V(G)=\{v_{1},v_{2},\ldots,v_{n}\}$.
Define an auxiliary graph $H$ as follows:
\begin{center}
$V(H)=U\cup W$, where
\end{center}
\begin{center}
$U=\{u_{1},u_{2},\ldots,u_{n}\}$, $W=\{w_{1},w_{2},\ldots,w_{n}\}$
and
\end{center}
\begin{center}
$E(H)=\left\{u_{i}w_{j},u_{j}w_{i}|~v_{i}v_{j}\in E(G), 1\leq i\leq
n,1\leq j\leq n\right\}\cup \{u_{i}w_{i}|~1\leq i\leq n\}$.
\end{center}

Clearly, $H$ is a bipartite graph with $\vert V(H)\vert = 2\vert
V(G)\vert$.

\begin{theorem}
\label{mytheorem3} If $\alpha$ is an interval total $t$-coloring of
the graph $G$, then there is an interval $t$-coloring $\beta$ of the
bipartite graph $H$.
\end{theorem}
\begin{proof} For the proof, we define an edge-coloring $\beta$ of the graph $H$
as follows:
\begin{description}
\item[(1)] $\beta (u_{i}w_{j})=\beta (u_{j}w_{i})=\alpha(v_{i}v_{j})$ for
every edge $v_{i}v_{j}\in E(G)$,

\item[(2)] $\beta (u_{i}w_{i})=\alpha(v_{i})$ for
$i=1,2,\ldots,n$.
\end{description}

It is easy to see that $\beta$ is an interval $t$-coloring of the
graph $H$. $~\square$
\end{proof}

This theorem shows that any interval total $t$-coloring of a graph
$G$ can be transform into an interval $t$-coloring of the bipartite
graph $H$.

\begin{corollary}
\label{mycorollary1} If $G$ is a connected graph and $G\in
\mathfrak{T}$, then
\begin{center}
$W_{\tau}(G)\leq 2\vert V(G)\vert -1$.
\end{center}
\end{corollary}
\begin{proof} Let $\alpha$ be an interval total $W_{\tau}(G)$-coloring of the graph $G$.
By Theorem \ref{mytheorem3}, $\beta$ is an interval
$W_{\tau}(G)$-coloring of the graph $H$. Since $H$ is a connected
bipartite graph with $\vert V(H)\vert = 2\vert V(G)\vert$ and $H\in
\mathfrak{N}$, by Theorem \ref{mytheorem1}, we have
\begin{center}
$W_{\tau}(G)\leq \vert V(H)\vert -1 = 2\vert V(G)\vert-1$, thus
\end{center}
\begin{center}
$W_{\tau}(G)\leq 2\vert V(G)\vert-1$.
\end{center}
$~\square$
\end{proof}

\begin{remark}
\label{myremark1}
Note that the upper bound in Corollary
\ref{mycorollary1} is sharp for simple paths $P_{n}$, since
$W_{\tau}(P_{n})=2n-1$ for any $n\in \mathbf{N}$.
\end{remark}

\begin{corollary}
\label{mycorollary2} If $G$ is a connected $r$-regular graph with
$\vert V(G)\vert \geq 2r+2$ and $G\in \mathfrak{T}$, then
\begin{center}
$W_{\tau}(G)\leq 2\vert V(G)\vert -3$.
\end{center}
\end{corollary}
\begin{proof} Let $\alpha$ be an interval total $W_{\tau}(G)$-coloring of the graph $G$.
By Theorem \ref{mytheorem3}, $\beta$ is an interval
$W_{\tau}(G)$-coloring of the graph $H$. Since $H$ is a connected
$(r+1)$-regular bipartite graph with $\vert V(H)\vert \geq 2(2r+2)$
and $H\in \mathfrak{N}$, by Theorem \ref{mytheorem2}, we have
\begin{center}
$W_{\tau}(G)\leq \vert V(H)\vert -3 = 2\vert V(G)\vert-3$, thus
\end{center}
\begin{center}
$W_{\tau}(G)\leq 2\vert V(G)\vert-3$.
\end{center}
$~\square$
\end{proof}

Next we derive some upper bounds for $W_{\tau}(G)$ depending on
degrees and diameter of a connected graph $G$.

\begin{theorem}
\label{mytheorem4}If $G$ is a connected graph and $G\in
\mathfrak{T}$, then
\begin{center}
$W_{\tau}(G)\leq 1+{\max\limits_{P\in \mathbf{P}}
}{\sum\limits_{v\in V(P)} }\ d_{G}(v)$,
\end{center}
where $\mathbf{P}$ is the set of all shortest paths in the graph
$G$.
\end{theorem}
\begin{proof} Consider an interval total $W_{\tau}(G)$-coloring
$\alpha$ of $G$. We distinguish four possible cases.

Case 1: there are vertices $v,v^{\prime}\in V(G)$ such that
$\alpha(v)=1$, $\alpha(v^{\prime})=W_{\tau}(G)$.

Let $P_{1}$ be a shortest path joining $v$ with $v^{\prime}$, where

\begin{center}
$P_{1}=\left( v_{1},e_{1},v_{2},e_{2},\ldots
,v_{i},e_{i},v_{i+1},\ldots ,v_{k},e_{k},v_{k+1}\right)$, $v_{1}=v$,
$v_{k+1}=v^{\prime}$.
\end{center}

Note that

\begin{center}

$\alpha (e_{1})\leq 1+ d_{G}(v_{1})$,

$\alpha (e_{2})\leq \alpha (e_{1})+d_{G}(v_{2})$,

$\cdots \cdots \cdots \cdots \cdots$

$\alpha (e_{i})\leq \alpha (e_{i-1})+d_{G}(v_{i})$,

$\cdots \cdots \cdots \cdots \cdots$

$\alpha (e_{k})\leq \alpha (e_{k-1})+d_{G}(v_{k})$,

$W_{\tau}(G)=\alpha (v^{\prime})=\alpha (v_{k+1}) \leq \alpha
(e_{k})+d_{G}(v_{k+1})$.
\end{center}

By summing these inequalities, we obtain

\begin{center}
$W_{\tau}(G)\leq $ $1+{{\sum\limits_{i=1}^{k+1}d_{G}(v_{i})}}\leq
1+{\max\limits_{P\in \mathbf{P}}}{\sum\limits_{v\in V(P)}}\
d_{G}(v)$.
\end{center}
Case 2: there is a vertex $v$ and there is an edge $e^{\prime}$ such
that $\alpha(v)=1$, $\alpha(e^{\prime})=W_{\tau}(G)$.

Let $e^{\prime}=v^{\prime}w$ and $P_{2}$ be a shortest path joining
$v$ with $v^{\prime}$, where

\begin{center}
$P_{2}=\left( v_{1},e_{1},v_{2},e_{2},\ldots
,v_{i},e_{i},v_{i+1},\ldots ,v_{k},e_{k},v_{k+1}\right)$, $v_{1}=v$,
$v_{k+1}=v^{\prime}$.
\end{center}

Note that

\begin{center}

$\alpha (e_{1})\leq 1+ d_{G}(v_{1})$,

$\alpha (e_{2})\leq \alpha (e_{1})+d_{G}(v_{2})$,

$\cdots \cdots \cdots \cdots \cdots$

$\alpha (e_{i})\leq \alpha (e_{i-1})+d_{G}(v_{i})$,

$\cdots \cdots \cdots \cdots \cdots$

$\alpha (e_{k})\leq \alpha (e_{k-1})+d_{G}(v_{k})$,

$W_{\tau}(G)=\alpha (e^{\prime})=\alpha (v_{k+1}w) \leq \alpha
(e_{k})+d_{G}(v_{k+1})$.
\end{center}

By summing these inequalities, we obtain

\begin{center}
$W_{\tau}(G)\leq $ $1+{{\sum\limits_{i=1}^{k+1}d_{G}(v_{i})}}\leq
1+{\max\limits_{P\in \mathbf{P}}}{\sum\limits_{v\in V(P)}}\
d_{G}(v)$.
\end{center}
Case 3: there is an edge $e$ and there is a vertex $v^{\prime}$ such
that $\alpha(e)=1$, $\alpha(v^{\prime})=W_{\tau}(G)$.

Let $e=uv$ and $P_{3}$ be a shortest path joining $v$ with
$v^{\prime}$, where

\begin{center}
$P_{3}=\left( v_{1},e_{1},v_{2},e_{2},\ldots
,v_{i},e_{i},v_{i+1},\ldots ,v_{k},e_{k},v_{k+1}\right)$, $v_{1}=v$,
$v_{k+1}=v^{\prime}$.
\end{center}

Note that

\begin{center}

$\alpha (e_{1})\leq 1+ d_{G}(v_{1})$,

$\alpha (e_{2})\leq \alpha (e_{1})+d_{G}(v_{2})$,

$\cdots \cdots \cdots \cdots \cdots$

$\alpha (e_{i})\leq \alpha (e_{i-1})+d_{G}(v_{i})$,

$\cdots \cdots \cdots \cdots \cdots$

$\alpha (e_{k})\leq \alpha (e_{k-1})+d_{G}(v_{k})$,

$W_{\tau}(G)=\alpha (v^{\prime})=\alpha (v_{k+1}) \leq \alpha
(e_{k})+d_{G}(v_{k+1})$.
\end{center}

By summing these inequalities, we obtain

\begin{center}
$W_{\tau}(G)\leq $ $1+{{\sum\limits_{i=1}^{k+1}d_{G}(v_{i})}}\leq
1+{\max\limits_{P\in \mathbf{P}}}{\sum\limits_{v\in V(P)}}\
d_{G}(v)$.
\end{center}
Case 4: there are edges $e,e^{\prime}\in E(G)$ such that
$\alpha(e)=1$, $\alpha(e^{\prime})=W_{\tau}(G)$.

Let $e=uv$ and $e^{\prime}=v^{\prime}w$. Without loss of generality
we may assume that a shortest path $P_{4}$ joining $e$ and
$e^{\prime}$ joins $v$ and $v^{\prime}$, where

\begin{center}
$P_{4}=\left( v_{1},e_{1},v_{2},e_{2},\ldots
,v_{i},e_{i},v_{i+1},\ldots ,v_{k},e_{k},v_{k+1}\right)$, $v_{1}=v$,
$v_{k+1}=v^{\prime}$.
\end{center}

Note that

\begin{center}

$\alpha (e_{1})\leq 1+ d_{G}(v_{1})$,

$\alpha (e_{2})\leq \alpha (e_{1})+d_{G}(v_{2})$,

$\cdots \cdots \cdots \cdots \cdots$

$\alpha (e_{i})\leq \alpha (e_{i-1})+d_{G}(v_{i})$,

$\cdots \cdots \cdots \cdots \cdots$

$\alpha (e_{k})\leq \alpha (e_{k-1})+d_{G}(v_{k})$,

$W_{\tau}(G)=\alpha (e^{\prime})=\alpha (v^{\prime}w)\leq \alpha
(e_{k})+d_{G}(v_{k+1})$.
\end{center}

By summing these inequalities, we obtain

\begin{center}
$W_{\tau}(G)\leq $ $1+{{\sum\limits_{i=1}^{k+1}d_{G}(v_{i})}}\leq
1+{\max\limits_{P\in \mathbf{P}}}{\sum\limits_{v\in V(P)}}\
d_{G}(v)$.
\end{center}
$~\square$
\end{proof}

\begin{corollary}
\label{mycorollary3} If $G$ is a connected graph and $G\in
\mathfrak{T}$, then $W_{\tau}(G)\leq 1+(diam(G)+1)\Delta(G)$.
\end{corollary}

Now we give an upper bound on $W_{\tau}(G)$ for graphs with a unique
universal vertex.

\begin{theorem}
\label{mytheorem5} If $G$ is a graph with a unique universal vertex
$u$ and $G\in \mathfrak{T}$, then $W_{\tau}(G)\leq \vert V(G)\vert +
2k(G)$, where $k(G)={\max }_{v\in V(G)(v\neq u)} d_{G}(v)$.
\end{theorem}
\begin{proof} Let $\alpha$ be an interval total
$W_{\tau}(G)$-coloring of the graph $G$.

Consider the vertex $u$. We show that $1\leq \min S[u,\alpha]\leq
k(G)+1$.

Suppose, to the contrary, that $\min S[u,\alpha]\geq k(G)+2$. Since
$d_{G}(v)\leq k(G)$ for any $v\in V(G)(v\neq u)$, then $\min
S[v,\alpha]\geq 2$ for any $v\in V(G)(v\neq u)$, which is a
contradiction.

Now, we have

\begin{center}
$1\leq \min S[u,\alpha]\leq k(G)+1$,
\end{center}
hence,
\begin{center}
$\vert V(G)\vert \leq \max S[u,\alpha]\leq \vert V(G)\vert + k(G)$.
\end{center}
This implies that $\max S[v,\alpha]\leq \vert V(G)\vert + 2k(G)$ for
any $v\in V(G)(v\neq u)$. $~\square$
\end{proof}

In the next theorem we prove that regular bipartite graphs, trees
and complete bipartite graphs are interval total colorable.

\begin{theorem}
\label{mytheorem6} The set $\mathfrak{T}$ contains all regular
bipartite graphs, trees and complete bipartite graphs.
\end{theorem}
\begin{proof} First we prove that if $G$ is an $r$-regular bipartite
graph with bipartition $(U,V)$, then $G$ has an interval total
$(r+2)$-coloring.

Since $G$ is an $r$-regular bipartite graph, we have $\chi^{\prime
}\left(G\right)=\Delta(G)=r$. Let $\alpha$ be a proper edge-coloring
of $G$ with colors $2,3,\ldots,r+1$. Clearly, $S(w,\alpha)=[2,r+1]$
for each $w\in V(G)$.

Define a total coloring $\beta$ of the graph $G$ as follows:

1. for any $u\in U$, let $\beta(u)=1$;

2. for any $e\in E(G)$, let $\beta(e)=\alpha(e)$;

3. for any $v\in V$, let $\beta(v)=r+2$.

It is easy to see that $\beta$ is an interval total $(r+2)$-coloring
of $G$.

Next we consider trees. Clearly, $K_{1}$ is a tree and has an
interval total $1$-coloring. Assume that $T$ is a tree and $T\neq
K_{1}$. Now we prove that $T$ has an interval total
$(\Delta(T)+2)$-coloring.

We use induction on $\vert E(T)\vert$. Clearly, the statement is
true for the case $\vert E(T)\vert=1$. Suppose that $\vert
E(T)\vert=k>1$ and the statement is true for all trees $T^{\prime}$
with $\vert E(T^{\prime})\vert<k$.

Suppose $e=uv\in E(T)$ and $d_{T}(u)=1$. Let $T^{\prime}=T-u$. Since
$\vert E(T)\vert>1$, we have $d_{T}(v)\geq 2$. Clearly,
$d_{T^{\prime}}(v)=d_{T}(v)-1, \Delta(T^{\prime})\leq \Delta(T)$ and
$\vert E(T^{\prime})\vert=\vert E(T)\vert-1<k$. Let $\alpha$ be an
interval total $(\Delta(T^{\prime})+2)$-coloring of the tree
$T^{\prime}$ (by induction hypothesis). Consider the vertex $v$. Let
\begin{center}
$S[v,\alpha]=\{s(1),s(2),\ldots,s(d_{T^{\prime}}(v)+1)\}$,
\end{center}
where $1\leq s(1)<s(2)<\ldots<s(d_{T^{\prime}}(v)+1)\leq
\Delta(T)+2$. We consider three cases.

Case 1: $s(1)=1$.

Clearly, $s(d_{T^{\prime}}(v)+1)=d_{T^{\prime}}(v)+1=d_{T}(v)$. In
this case we color the edge $e$ with color $d_{T}(v)+1$ and the
vertex $u$ with color $d_{T}(v)+2$. It is easy to see that the
obtained coloring is an interval total $(\Delta(T)+2)$-coloring of
the tree $T$.

Case 2: $s(1)=2$.

Subcase 2.1: $\alpha(v)=2$.

Clearly, $s(d_{T^{\prime}}(v)+1)=d_{T}(v)+1$. In this case we color
the edge $e$ with color $d_{T}(v)+2$ and the vertex $u$ with color
$d_{T}(v)+1$. It is easy to see that the obtained coloring is an
interval total $(\Delta(T)+2)$-coloring of the tree $T$.

Subcase 2.2: $\alpha(v)\neq 2$ and $\Delta (T^{\prime})=\Delta (T)$.

We color the edge $e$ with color $1$ and the vertex $u$ with color
$2$. It is easy to see that obtained coloring is an interval total
$(\Delta(T)+2)$-coloring of the tree $T$.

Subcase 2.3: $\alpha(v)\neq 2$ and $\Delta (T^{\prime})<\Delta (T)$.

We define a total coloring $\beta$ of the tree $T^{\prime}$ in the
following way:

1. $\forall w\in V(T^{\prime})$   $\beta(w)=\alpha(w)+1$;

2. $\forall e^{\prime}\in E(T^{\prime})$
$\beta(e^{\prime})=\alpha(e^{\prime})+1$.

Now we color the edge $e$ with color $2$ and the vertex $u$ with
color $1$. It is not difficult to see that the obtained coloring is
an interval total $(\Delta (T)+2)$-coloring of the tree $T$.

Case 3: $s(1)\geq 3$.

We color the edge $e$ with color $s(1)-1$ and the vertex $u$ with
color $s(1)-2$. It is easy to see that the obtained coloring is an
interval total $(\Delta (T)+2)$-coloring of the tree $T$.

Finally, we prove that if $K_{m,n}$ is a complete bipartite graph,
then it has an interval total $(m+n+1)$-coloring.

Let
$V(K_{m,n})=\{u_{1},u_{2},\ldots,u_{m},v_{1},v_{2},\ldots,v_{n}\}$
and $E(K_{m,n})=\{u_{i}v_{j}|~1\leq i\leq m,1\leq j\leq n\}$.

Define a total coloring $\gamma$ of the graph $K_{m,n}$ as follows:

1. for $i=1,2,\ldots,m$, let $\gamma(u_{i})=i$;

2. for $j=1,2,\ldots,n$, let $\gamma(v_{j})=m+1+j$;

3. for $i=1,2,\ldots,m$ and $j=1,2,\ldots,n$, let
$\gamma(u_{i}v_{j})=i+j$.

It is easy to see that $\gamma$ is an interval total
$(m+n+1)$-coloring of $K_{m,n}$. $~\square$
\end{proof}

\begin{corollary}
\label{mycorollary4} If $G$ is an $r$-regular bipartite graph, then
$w_{\tau}(G)\leq r+2$.
\end{corollary}

\begin{corollary}
\label{mycorollary5} If $T$ is a tree, then $w_{\tau}(T)\leq
\Delta(T)+2$.
\end{corollary}

\begin{corollary}
\label{mycorollary6} $W_{\tau}(K_{m,n})\geq m+n+1$ for any $m,n\in
\mathbf{N}$.
\end{corollary}

From Corollary \ref{mycorollary4}, we have that $w_{\tau}(G)\leq
r+2$ for any $r$-regular bipartite graph $G$. On the other hand,
clearly, $w_{\tau}(G)\geq r+1$. In \cite{b14,b19} it was proved that
the problem of determining whether $\chi^{\prime
\prime}\left(G\right)=r+1$ is $NP$-complete even for cubic bipartite
graphs. Thus, we can conclude that the verification whether
$w_{\tau}(G)=r+1$ for any $r$-regular ($r\geq 3$) bipartite graph
$G$ is also $NP$-complete.
\bigskip

\section{Exact values of $w_{\tau}(G)$ and $W_{\tau}(G)$}\

In this section we determine the exact values of $w_{\tau}(G)$ and
$W_{\tau}(G)$ for simple cycles, complete graphs and wheels.

In \cite{b25} it was proved the following result.

\begin{theorem}
\label{mytheorem7} For the simple cycle $C_{n}$,
\begin{center}
$\chi^{\prime\prime}(C_{n})=\left\{
\begin{tabular}{ll}
$3$, & if $n=3k$, \\
$4$, & if $n\neq 3k$. \\
\end{tabular}%
\right.$
\end{center}
\end{theorem}

\begin{theorem}
\label{mytheorem8} For any $n\geq 3$, we have
\begin{description}
\item[(1)] $C_{n}\in \mathfrak{T}$,

\item[(2)] $w_{\tau}(C_{n})=\left\{
\begin{tabular}{ll}
$3$, & if $n=3k$, \\
$4$, & if $n\neq 3k$,\\
\end{tabular}%
\right.$

\item[(3)] $W_{\tau}(C_{n})=n+2$.
\end{description}
\end{theorem}
\begin{proof} First we prove that $C_{n}$ has either an interval total $3$-coloring or an interval total $4$-coloring.

Let $V(C_{n})=\{v_{1},v_{2}\ldots,v_{n}\}$ and
$E(C_{n})=\{v_{i}v_{i+1}|~1\leq i\leq n-1\}\cup \{v_{1}v_{n}\}$. We
consider three cases.

Case 1: $n=3k$ ($k\in \mathbf{N}$).

Define a total coloring $\alpha$ of the graph $C_{n}$ as follows:

\begin{center}
$\alpha \left(v_{i}\right) =\left\{
\begin{tabular}{ll}
$2$, & if $i\equiv 0 \pmod{3}$, \\
$1$, & if $i\equiv 1 \pmod{3}$, \\
$3$, & if $i\equiv 2 \pmod{3}$, \\
\end{tabular}%
\right.$
\end{center}
for $i=1,2,\ldots,n$,

\begin{center}
$\alpha \left(v_{j}v_{j+1}\right) =\left\{
\begin{tabular}{ll}
$3$, & if $j\equiv 0 \pmod{3}$, \\
$2$, & if $j\equiv 1 \pmod{3}$, \\
$1$, & if $j\equiv 2 \pmod{3}$, \\
\end{tabular}%
\right.$
\end{center}
for $j=1,2,\ldots,n-1$, and $\alpha(v_{1}v_{n})=3$.

Case 2: $n\neq 3k$ ($k\in \mathbf{N}$) and $n$ is even.

Define a total coloring $\alpha$ of the graph $C_{n}$ as follows:

\begin{center}
$\alpha \left(v_{i}\right) =\left\{
\begin{tabular}{ll}
$4$, & if $i\equiv 0 \pmod{2}$, \\
$1$, & if $i\equiv 1 \pmod{2}$, \\
\end{tabular}%
\right.$
\end{center}
for $i=1,2,\ldots,n$,

\begin{center}
$\alpha \left(v_{j}v_{j+1}\right) =\left\{
\begin{tabular}{ll}
$2$, & if $j\equiv 0 \pmod{2}$, \\
$3$, & if $j\equiv 1 \pmod{2}$, \\
\end{tabular}%
\right.$
\end{center}
for $j=1,2,\ldots,n-1$, and $\alpha(v_{1}v_{n})=2$.

Case 3: $n\neq 3k$ ($k\in \mathbf{N}$) and $n$ is odd.

Define a total coloring $\alpha$ of the graph $C_{n}$ as follows:

\begin{center}
$\alpha \left(v_{i}\right) =\left\{
\begin{tabular}{ll}
$4$, & if $i\equiv 0 \pmod{2}$, $i\neq n-1$, \\
$1$, & if $i\equiv 1 \pmod{2}$, $i\neq n$, \\
$2$, & if $i=n-1$, \\
$3$, & if $i=n$, \\
\end{tabular}%
\right.$
\end{center}
for $i=1,2,\ldots,n$,

\begin{center}
$\alpha \left(v_{j}v_{j+1}\right) =\left\{
\begin{tabular}{ll}
$2$, & if $j\equiv 0 \pmod{2}$, $j\neq n-1$, \\
$3$, & if $j\equiv 1 \pmod{2}$,  \\
$4$, & if $j=n-1$, \\
\end{tabular}%
\right.$
\end{center}
for $j=1,2,\ldots,n-1$, and $\alpha(v_{1}v_{n})=2$.

It is easy to check that $\alpha$ is an interval total $3$-coloring
of the graph $C_{n}$, when $n=3k$, and an interval total
$4$-coloring of the graph $C_{n}$, when $n\neq 3k$. Hence, for any
$n\geq 3$, $C_{n}\in \mathfrak{T}$ and $w_{\tau}(C_{n})\leq 3$ if
$n=3k$ and $w_{\tau}(C_{n})\leq 4$ if $n\neq 3k$. On the other hand,
by Theorem \ref{mytheorem7}, and taking into account that
$w_{\tau}(C_{n})\geq \chi^{\prime\prime}(C_{n})$, we have
$w_{\tau}(C_{n})\geq 3$ if $n=3k$ and $w_{\tau}(C_{n})\geq 4$ if
$n\neq 3k$. Thus, (1) and (2) hold.

Let us prove (3).

Now we show that $W_{\tau}(C_{n})\geq n+2$ for any $n\geq 3$. For
that, we consider two cases.

Case 1: $n$ is even.

Define a total coloring $\beta$ of the graph $C_{n}$ as follows:

1. for $i=1,2,\ldots,\frac{n}{2}$, let
\begin{center}
$\beta(v_{i})=2i-1$, $\beta(v_{i}v_{i+1})=2i$
\end{center}

2. for $j=\frac{n}{2}+1,\ldots,n$, let
\begin{center}
$\beta(v_{j})=2(n-j)+4$,
\end{center}

3. for $k=\frac{n}{2}+1,\ldots,n-1$, let
\begin{center}
$\beta(v_{k}v_{k+1})=2(n-k)+3$,
\end{center}
and $\beta(v_{1}v_{n})=3$.

Case 2: $n$ is odd.

Define a total coloring $\beta$ of the graph $C_{n}$ as follows:

1. for $i=1,2,\ldots,\lceil\frac{n}{2}\rceil+1$, let
\begin{center}
$\beta(v_{i})=2i-1,$
\end{center}

2. for $j=\lceil\frac{n}{2}\rceil+2,\ldots,n$, let
\begin{center}
$\beta(v_{j})=2(n-j)+4$,
\end{center}

3. for $k=1,2,\ldots,\lceil\frac{n}{2}\rceil$, let
\begin{center}
$\beta(v_{k}v_{k+1})=2k$,
\end{center}

4. for $l=\lceil\frac{n}{2}\rceil+1,\ldots,n-1$, let
\begin{center}
$\beta(v_{l}v_{l+1})=2(n-l)+3$,
\end{center}
and $\beta(v_{1}v_{n})=3$.

It is not difficult to see that $\beta$ is an interval total
$(n+2)$-coloring of the graph $C_{n}$. Thus, $W_{\tau}(C_{n})\geq
n+2$ for any $n\geq 3$. On the other hand, using Corollary
\ref{mycorollary3}, and taking into account that
$diam(C_{n})=\lfloor\frac{n}{2}\rfloor$ and $\Delta(C_{n})=2$, it is
easy to show that $W_{\tau}(C_{n})\leq n+2$ for any $n\geq 3$.
$~\square$
\end{proof}

In \cite{b5} it was proved the following result.

\begin{theorem}
\label{mytheorem9} For the complete graph $K_{n}$,
\begin{center}
$\chi^{\prime\prime}(K_{n})=\left\{
\begin{tabular}{ll}
$n$, & if $n$ is odd, \\
$n+1$, & if $n$ is even. \\
\end{tabular}%
\right.$
\end{center}
\end{theorem}

\begin{theorem}
\label{mytheorem10} For any $n\in \mathbf{N}$, we have
\begin{description}
\item[(1)] $K_{n}\in \mathfrak{T}$,

\item[(2)] $w_{\tau}(K_{n})=\left\{
\begin{tabular}{ll}
$n$, & if $n$ is odd, \\
$\frac{3}{2}n$, & if $n$ is even,\\
\end{tabular}%
\right.$

\item[(3)] $W_{\tau}(K_{n})=2n-1$.
\end{description}
\end{theorem}
\begin{proof} Let $V(K_{n})=\{v_{1},v_{2},\ldots,v_{n}\}$.

First we show that $K_{n}$ has an interval total $(2n-1)$-coloring
for any $n\in \mathbf{N}$. For that, we define a total coloring
$\alpha$ of the graph $K_{n}$ as follows:

1. for $i=1,2,\ldots,n$, let $\alpha(v_{i})=2i-1$;

2. for $i=1,2,\ldots,n$ and $j=1,2,\ldots,n$, where $i\neq j$, let
$\alpha(v_{i}v_{j})=i+j-1$.

It is easy to see that $\alpha$ is an interval total
$(2n-1)$-coloring of the graph $K_{n}$. This proves that $K_{n}\in
\mathfrak{T}$ and $W_{\tau}(K_{n})\geq 2n-1$ for any $n\in
\mathbf{N}$. On the other hand, using Corollary \ref{mycorollary3},
and taking into account that $diam(K_{n})=1$ and
$\Delta(K_{n})=n-1$, it is simple to show that $W_{\tau}(K_{n})\leq
2n-1$ for any $n\in \mathbf{N}$. Thus, (1) and (3) hold.

Let us prove (2). We  consider two cases.

Case 1: $n$ is odd.

Since $K_{n}$ is a regular graph, by Theorem \ref{mytheorem9}, we
have $w_{\tau}(K_{n})=\chi^{\prime\prime}(K_{n})=n$.

Case 2: $n$ is even.

Now we show that $w_{\tau}(K_{n})\leq \frac{3}{2}n$.

Define a total coloring $\beta$ of the graph $K_{n}$ as follows:

1. for $i=1,2,\ldots,\frac{n}{2}$, let
\begin{center}
$\beta(v_{i})=i$;
\end{center}

2. for $j=\frac{n}{2}+1,\ldots,n$, let
\begin{center}
$\beta(v_{j})=\frac{n}{2}+j$;
\end{center}

3. for $i=1,2,\ldots,n$, $j=1,2,\ldots,n$, $i<j$, $i+j$ is odd, and
$i+j-1\leq n$, let
\begin{center}
$\beta(v_{i}v_{j})=\frac{n}{2}+\frac{i+j-1}{2}$;
\end{center}

4. for $i=1,2,\ldots,n$, $j=1,2,\ldots,n$, $i<j$, $i+j$ is odd, and
$i+j-1> n$, let
\begin{center}
$\beta(v_{i}v_{j})=\frac{i+j-1}{2}$;
\end{center}

5. for $i=1,2,\ldots,n$, $j=1,2,\ldots,n$, $i<j$, $i+j$ is even, and
$i+j\leq n$, let
\begin{center}
$\beta(v_{i}v_{j})=\frac{i+j}{2}$;
\end{center}

6. for $i=1,2,\ldots,n$, $j=1,2,\ldots,n$, $i<j$, $i+j$ is even, and
$i+j> n$, let
\begin{center}
$\beta(v_{i}v_{j})=\frac{n}{2}+\frac{i+j}{2}$.
\end{center}

Let us show that $\beta$ is an interval total
$\frac{3}{2}n$-coloring of the graph $K_{n}$.

Let $v_{i}\in V(K_{n})$, where $1\leq i\leq n$.

If $i$ is even, by the definition of $\beta$, we have

\begin{eqnarray*}
S\left[v_{i},\beta \right] &=&\left({\bigcup\limits_{1\leq l\leq
\frac{n+2-i}{2}}\left\{\frac{n}{2}+\frac{i+(2l-1)-1}{2}\right\}}\right)\cup
\left({\bigcup\limits_{\frac{n+2-i}{2}<l\leq
\frac{n}{2}}\left\{\frac{i+(2l-1)-1}{2}\right\}}\right)\cup\\
&& \left({\bigcup\limits_{1\leq l\leq \frac{n-i}{2}, l\neq
\frac{i}{2}}\left\{\frac{i+2l}{2}\right\}}\right)\cup
\left({\bigcup\limits_{\frac{n-i}{2}<l\leq \frac{n}{2}, l\neq
\frac{i}{2}}\left\{\frac{n}{2}+\frac{i+2l}{2}\right\}}\right)\cup
\left\{i+\frac{n}{2}sg\left(i-\frac{n}{2}\right)\right\}\\
&=& \left[\frac{n+i}{2},n\right]\cup
\left[\frac{n}{2}+1,\frac{n+i}{2}-1\right]\cup
\left(\left[\frac{i}{2}+1,\frac{n}{2}\right]\setminus
\{i\}\right)\cup\\
&& \left(\left[n+1,\frac{i}{2}+n\right]\setminus
\left\{\frac{n}{2}+i\right\}\right)\cup
\left\{i+\frac{n}{2}sg\left(i-\frac{n}{2}\right)\right\}\\
&=& \left[\frac{i}{2}+1,\frac{i}{2}+n\right],
\end{eqnarray*}

and if $i$ is odd, by the definition of $\beta$, we have

\begin{eqnarray*}
S\left[v_{i},\beta \right] &=&\left({\bigcup\limits_{1\leq l\leq
\frac{n+1-i}{2}}\left\{\frac{n}{2}+\frac{i+2l-1}{2}\right\}}\right)\cup
\left({\bigcup\limits_{\frac{n+1-i}{2}<l\leq
\frac{n}{2}}\left\{\frac{i+2l-1}{2}\right\}}\right)\cup\\
&& \left({\bigcup\limits_{1\leq l\leq \frac{n+1-i}{2}, l\neq
\frac{i+1}{2}}\left\{\frac{i+2l-1}{2}\right\}}\right)\cup
\left({\bigcup\limits_{\frac{n+1-i}{2}<l\leq \frac{n}{2}, l\neq
\frac{i+1}{2}}\left\{\frac{n}{2}+\frac{i+2l-1}{2}\right\}}\right)\cup\\
&& \left\{i+\frac{n}{2}sg\left(i-\frac{n}{2}\right)\right\}\\
&=& \left[\frac{n+i+1}{2},n\right]\cup
\left[\frac{n}{2}+1,\frac{n+i-1}{2}\right]\cup
\left(\left[\frac{i+1}{2},\frac{n}{2}\right]\setminus
\{i\}\right)\cup\\
&& \left(\left[n+1,\frac{i-1}{2}+n\right]\setminus
\left\{\frac{n}{2}+i\right\}\right)\cup
\left\{i+\frac{n}{2}sg\left(i-\frac{n}{2}\right)\right\}\\
&=& \left[\frac{i+1}{2},\frac{i-1}{2}+n\right].
\end{eqnarray*}

This shows that $\beta$ is an interval total $\frac{3}{2}n$-coloring
of the graph $K_{n}$.

Next we prove that $w_{\tau}(K_{n})\geq \frac{3}{2}n$.

Suppose, to the contrary, that $\gamma$ is an interval total
$w_{\tau}(K_{n})$-coloring of the graph $K_{n}$, where $n\leq
w_{\tau}(K_{n})\leq \frac{3}{2}n-1$.

Since $w_{\tau}(K_{n})\geq \chi^{\prime\prime}(K_{n})$, by Theorem
\ref{mytheorem9}, we have $n+1\leq w_{\tau}(K_{n})\leq
\frac{3}{2}n-1$.

Consider the vertices $v_{1},v_{2},\ldots,v_{n}$. It is clear that
\begin{center}
$1\leq \min S[v_{i},\gamma]\leq w_{\tau}(K_{n})-n+1$ for
$i=1,2,\ldots,n$.
\end{center}

Hence, $\{w_{\tau}(K_{n})-n+1,\ldots,n\}\subseteq S[v_{i},\gamma]$
for $i=1,2,\ldots,n$. Let us show that none of the vertices
$v_{1},v_{2},\ldots,v_{n}$ is colored by $j$,
$j=w_{\tau}(K_{n})-n+1,\ldots,n$. Suppose that
$\gamma(v_{i_{0}})=j_{0}$, $j_{0}\in
\{w_{\tau}(K_{n})-n+1,\ldots,n\}$. Clearly, $\gamma(v_{i})\neq
j_{0}$ for $i=1,2,\ldots,n$ and $i\neq i_{0}$. This implies that any
vertex $v_{i}$, except $v_{i_{0}}$, is incident to an edge with
color $j_{0}$, which is a contradiction. The contradiction shows
that $\gamma (v_{i})\notin \{w_{\tau}(K_{n})-n+1,\ldots,n\}$ for
$i=1,2,\ldots,n$. Hence,
\begin{center}
$\gamma (v_{i})\in \{1,2,\ldots,w_{\tau}(K_{n})-n\}\cup
\{n+1,\ldots,w_{\tau}(K_{n})\}$ for $i=1,2,\ldots,n$.
\end{center}

On the other hand, since $\chi(K_{n})=n$, we have
\begin{center}
$\vert\{1,2,\ldots,w_{\tau}(K_{n})-n\}\vert +
\vert\{n+1,\ldots,w_{\tau}(K_{n})\}\vert\geq n$,
\end{center}
thus $w_{\tau}(K_{n})\geq \frac{3}{2}n$, which is a contradiction.
$~\square$
\end{proof}

\begin{theorem}
\label{mytheorem11} For any $n\in \mathbf{N}$,
\begin{description}
\item[(1)] if $2n-1\leq t\leq 4n-3$, then $K_{2n-1}\in \mathfrak{T}_{t}$,

\item[(2)] if $3n\leq t\leq 4n-1$, then $K_{2n}\in \mathfrak{T}_{t}$.
\end{description}
\end{theorem}
\begin{proof}
First we prove (1). For that, we transform the interval total
$(4n-3)$-coloring $\alpha$ of the graph $K_{2n-1}$ constructed in
 the proof of Theorem \ref{mytheorem10}, into an interval total
$t$-coloring $\beta$ of the same graph.

For every $v\in V(K_{2n-1})$, we set:

\begin{center}
$\beta (v) =\left\{
\begin{tabular}{ll}
$\alpha(v)$, & if $1\leq \alpha(v)\leq t$, \\
$\alpha(v)-2n+1$, & if $t+1\leq \alpha(v)\leq 4n-3$. \\
\end{tabular}%
\right.$
\end{center}

For every $e\in E(K_{2n-1})$, we set:

\begin{center}
$\beta (e) =\left\{
\begin{tabular}{ll}
$\alpha(e)$, & if $1\leq \alpha(e)\leq t$, \\
$\alpha(e)-2n+1$, & if $t+1\leq \alpha(e)\leq 4n-3$. \\
\end{tabular}%
\right.$
\end{center}

It is easy to see that $\beta$ is an interval total $t$-coloring of
the graph $K_{2n-1}$.

Let us prove (2).

For that, we transform the interval total $3n$-coloring $\beta$ of
the graph $K_{2n}$ constructed in the proof of Theorem
\ref{mytheorem10}, into an interval total $t$-coloring $\gamma$ of
the same graph.

Define a total coloring $\gamma$ of the graph $K_{2n}$ as follows:

1. for $i=1,2,\ldots,2n$, let
\begin{center}
$\gamma (v_{i}) =\left\{
\begin{tabular}{ll}
$\beta(v_{i})+t-3n$, & if $\beta(v_{i})+t-3n\leq 2i-1$, \\
$2i-1$, & if $\beta(v_{i})+t-3n>2i-1$; \\
\end{tabular}%
\right.$
\end{center}

2. for $i=1,2,\ldots,2n-1$, $j=1,2,\ldots,2n-1$, $i\neq j$, and
$i+j-1\leq 2(t-3n)+1$, let
\begin{center}
$\gamma (v_{i}v_{j}) =i+j-1$;
\end{center}

3. for $i=1,2,\ldots,2n$, $j=1,2,\ldots,2n$, $i\neq j$, and
$2(t-3n)+1<i+j-1<2n$, let
\begin{center}
$\gamma (v_{i}v_{j}) =\left\{
\begin{tabular}{ll}
$\beta(v_{i}v_{j})+t-3n$, & if $i+j$ is even, \\
$\beta(v_{i}v_{j})$, & if $i+j$ is odd; \\
\end{tabular}%
\right.$
\end{center}

4. for $i=1,2,\ldots,2n$, $j=1,2,\ldots,2n$, $i\neq j$, and $2n\leq
i+j-1\leq 2n+2(t-3n)+1$, let
\begin{center}
$\gamma (v_{i}v_{j}) =i+j-1$;
\end{center}

5. for $i=3,4,\ldots,2n$, $j=3,4,\ldots,2n$, $i\neq j$, and
$i+j-1>2n+2(t-3n)+1$, let
\begin{center}
$\gamma (v_{i}v_{j}) =\left\{
\begin{tabular}{ll}
$\beta(v_{i}v_{j})+t-3n$, & if $i+j$ is even, \\
$\beta(v_{i}v_{j})$, & if $i+j$ is odd; \\
\end{tabular}%
\right.$
\end{center}

It can be easily verified that $\gamma$ is an interval total
$t$-coloring of the graph $K_{2n}$. $~\square$
\end{proof}

Finally, we obtain the exact values of $w_{\tau}(G)$ and
$W_{\tau}(G)$ for wheels. Recall that a wheel $W_{n}$ $(n\geq 4)$ is
defined as follows:
\begin{center}
$V(W_{n})=\left\{ u,v_{1},v_{2},\ldots,v_{n-1}\right\}$ and
$E(W_{n})= \left\{uv_{i}|~1\leq i\leq n-1\right\}\cup\left\{
v_{i}v_{i+1}|~1\leq i\leq n-2\right\}\cup \{v_{1}v_{n-1}\}$.
\end{center}

\begin{lemma}\label{mylemma1}
For any $n\geq 4$, we have $W_{n}\in \mathfrak{T}$ and
\begin{center}
$w_{\tau}(W_{n})=\left\{
\begin{tabular}{ll}
$n+2$, & if $n=4$, \\
$n$, & if $n\geq 5$. \\
\end{tabular}%
\right.$
\end{center}
\end{lemma}
\begin{proof}
Clearly, $W_{4}=K_{4}$, hence, by Theorem \ref{mytheorem10}, we have
$W_{4}\in \mathfrak{T}$ and $w_{\tau}(W_{4})=w_{\tau}(K_{4})=6$.

Assume that $n\geq 5$.

For the proof of the lemma we construct an interval total
$n$-coloring of the graph $W_{n}$. We consider two cases.

Case 1: $n$ is even.

Define a total coloring $\alpha $ of the graph $W_{n}$ as follows:

1) $\alpha(u)=n$, $\alpha(v_{1})=2$ and for
$i=2,\ldots,\frac{n}{2}-1$, let $\alpha \left( v_{i}\right) =2i+1$;

2) $\alpha(v_{\frac{n}{2}})=n-2$, $\alpha(v_{\frac{n}{2}+1})=n-4$,
and for $j=\frac{n}{2}+2,\ldots,n-1$, let $\alpha (v_{j})
=2(n-j+1);$

3) for $k=1,2,\ldots,\frac{n}{2}$, let
\begin{center}
$\alpha\left(uv_{k}\right) =2k-1$;
\end{center}

4) for $l=\frac{n}{2}+1,\ldots,n-1$, let
\begin{center}
$\alpha\left(uv_{l}\right) =2(n-l)$;
\end{center}

5) for $p=1,\ldots,\frac{n}{2}-1$, let

\begin{center}
$\alpha\left(v_{p}v_{p+1}\right) =2(p+1)$ and $\alpha
\left(v_{\frac{n}{2}}v_{\frac{n}{2}+1}\right) =n-3$;
\end{center}

6) for $q=\frac{n}{2}+1,\ldots,n-2$, let

\begin{center}
$\alpha\left(v_{q}v_{q+1}\right) =2(n-q)+1$ and
$\alpha\left(v_{1}v_{n-1}\right) =3$.
\end{center}

Case 2: $n$ is odd.

Define a total coloring $\beta $ of the graph $W_{n}$ as follows:

1) $\beta(u)=n$, $\beta(v_{1})=2$ and for $i=2,\ldots,\lfloor
\frac{n}{2}\rfloor-1$, let $\beta \left(v_{i}\right) =2i+1$;

2)$\beta(v_{\lfloor\frac{n}{2}\rfloor})=n-4$,
$\beta(v_{\lceil\frac{n}{2}\rceil})=n-2$ and for
$j=\lceil\frac{n}{2}\rceil+1,\ldots,n-1$, let $\beta(v_{j})
=2(n-j+1)$;

3) for $k=1,2,\ldots,\lfloor\frac{n}{2}\rfloor$, let
\begin{center}
$\beta\left(uv_{k}\right) =2k-1;$
\end{center}

4) for $l=\lceil\frac{n}{2}\rceil,\ldots,n-1$, let
\begin{center}
$\beta\left(uv_{l}\right) =2(n-l)$;
\end{center}

5) for $p=1,\ldots,\lfloor\frac{n}{2}\rfloor-1$, let
\begin{center}
$\beta\left(v_{p}v_{p+1}\right) =2(p+1)$ and $\beta
\left(v_{\lfloor\frac{n}{2}\rfloor}v_{\lceil\frac{n}{2}\rceil}\right)
=n-3$;
\end{center}

6) for $q=\lceil\frac{n}{2}\rceil,\ldots,n-2$, let
\begin{center}
$\beta\left(v_{q}v_{q+1}\right) =2(n-q)+1$ and $\beta \left(
v_{1}v_{n-1}\right) =3$.
\end{center}

It is not difficult to check that $\alpha$ is an interval total
$n$-coloring of the graph $W_{n}$, when $n$ is even, and $\beta$ is
an interval total $n$-coloring of the graph $W_{n}$, when $n$ is
odd. Hence, $W_{n}\in \mathfrak{T}$. On the other hand, clearly,
$w_{\tau}(W_{n})\geq \chi^{\prime\prime}(W_{n})=\Delta(W_{n})+1=n$,
thus $w_{\tau}(W_{n})=n$.
 $~\square$
\end{proof}

\begin{lemma}\label{mylemma2}
For any $n\geq 5$, we have $W_{n}\in
\mathfrak{T}_{n+1}\cap\mathfrak{T}_{n+2}$.
\end{lemma}
\begin{proof}
First we show that $W_{n}\in \mathfrak{T}_{n+2}$ for any $n\geq 5$.

Define a total coloring $\alpha $ of the graph $W_{n}$ as follows:

1) $\alpha(u)=1$,
 $\alpha(v_{1})=3$, $\alpha(v_{\lceil\frac{n}{2}\rceil}) =n-1$ and for
$i=2,\ldots,\lceil\frac{n}{2}\rceil-1$, let $\alpha\left(
v_{i}\right) =2(i+1)$;

2) for $j=\lceil\frac{n}{2}\rceil+1,\ldots,n-1$, let $\alpha(v_{j})
=2(n-j)+3$;

3) for $k=1,2,\ldots,\lfloor\frac{n}{2}\rfloor$, let
\begin{center}
$\alpha\left(uv_{k}\right) =2k$;
\end{center}
4) for $l=\lfloor\frac{n}{2}\rfloor+1,\ldots,n-1$, let
\begin{center}
$\alpha\left(uv_{l}\right) =2(n-l)+1$;
\end{center}
5) for $p=1,\ldots,\lfloor\frac{n-1}{2}\rfloor$, let
\begin{center}
$\alpha\left(v_{p}v_{p+1}\right) =2p+3$;
\end{center}
6) for $q=\lfloor\frac{n-1}{2}\rfloor+1,\ldots,n-2$
\begin{center}
$\alpha\left(v_{q}v_{q+1}\right)=2(n-q+1)$ and $\alpha\left(
v_{1}v_{n-1}\right) =4$.
\end{center}

It is easily seen that $\alpha$ is an interval total
$(n+2)$-coloring of the graph $W_{n}$.

Now we show that $W_{n}\in \mathfrak{T}_{n+1}$ for any $n\geq 5$.

Define a total coloring $\beta $ of the graph $W_{n}$ as follows:

1) for $\forall v\in V(W_{n})$, let $\beta(v)=\alpha(v)$;

2) for $\forall e\in E(W_{n})$, let

\begin{center}
$\beta (e)=\left\{
\begin{tabular}{ll}
$\alpha(e)$, & if $\alpha(e)\neq n+2$, \\
$n-2$, & otherwise. \\
\end{tabular}%
\right.$
\end{center}

It is easily seen that $\beta$ is an interval total $(n+1)$-coloring
of the graph $W_{n}$.
 $~\square$
\end{proof}

\begin{lemma}
\label{mylemma3} For any $n\geq 4$, we have $W_{\tau}(W_{n})\geq
n+3$.
\end{lemma}
\begin{proof}
Clearly, for the proof of the lemma it suffices to construct
an interval total $(n+3)$-coloring of the graph $W_{n}$ for $n\geq
4$. We consider two cases.

Case 1: $n$ is even.

Define a total coloring $\alpha$ of the graph $W_{n}$ as follows:

1) for $i=1,2,\ldots,\frac{n}{2}+1$, let $\alpha\left(v_{i}\right)
=2i-1$;

2) for $j=\frac{n}{2}+2,\ldots,n-1$, let $\alpha(v_{j}) =2(n-j+1)$;

3) for $k=1,2,\ldots,\frac{n}{2}$, let
\begin{center}
$\alpha\left(v_{k}v_{k+1}\right) =2k$;
\end{center}

4) for $l=\frac{n}{2}+1,\ldots,n-2$, let
\begin{center}
$\alpha\left(v_{l}v_{l+1}\right)=2(n-l)+1$ and $\alpha\left(
v_{1}v_{n-1}\right) =3$;
\end{center}

5) for $p=2,\ldots,\frac{n}{2}$, let
\begin{center}
$\alpha\left(uv_{p}\right)=2p+1$ and $\alpha\left(uv_{1}\right)=4$;
\end{center}

6) for $q=\frac{n}{2}+1,\ldots,n-1$, let
\begin{center}
$\alpha\left(uv_{q}\right)=2(n-q+2)$ and $\alpha(u)=n+3$.
\end{center}

Case 2: $n$ is odd.

Define a total coloring $\beta$ of the graph $W_{n}$ as follows:

1) for $i=1,2,\ldots,\lfloor\frac{n}{2}\rfloor$, let $\beta \left(
v_{i}\right) =2i-1$, $\beta\left(v_{i}v_{i+1}\right)=2i$;

2) for $j=\lceil\frac{n}{2}\rceil,\ldots,n-1$, let $\beta(v_{j})
=2(n-j+1)$;

3) for $k=\lceil\frac{n}{2}\rceil,\ldots,n-2$, let
\begin{center}
$\beta\left(v_{k}v_{k+1}\right)=2(n-k)+1$ and
$\beta\left(v_{1}v_{n-1}\right)=3$;
\end{center}

4) for $p=2,3,\ldots,\lceil\frac{n}{2}\rceil$, let
\begin{center}
$\beta\left(uv_{p}\right)=2p+1$ and $\beta\left(uv_{1}\right)=4$;
\end{center}

5) for $q=\lceil\frac{n}{2}\rceil+1,\ldots,n-1$, let
\begin{center}
$\beta\left(uv_{q}\right)=2(n-q+2)$ and $\beta(u) =n+3$.
\end{center}

It is not difficult to check that $\alpha$ is an interval total
$(n+3)$-coloring of the graph $W_{n}$, when $n$ is even, and $\beta
$ is an interval total $(n+3)$-coloring of the graph $W_{n}$, when
$n$ is odd.
 $~\square$
\end{proof}

\begin{remark}
\label{myremark2} Easy analysis shows that if $4\leq n\leq 8$, then
$W_{\tau}(W_{n})=n+3$.
\end{remark}

\begin{lemma}
\label{mylemma4} For any $n\geq 9$, we have $W_{\tau}(W_{n})\geq
n+4$.
\end{lemma}
\begin{proof} Clearly, for the proof of the lemma it suffices to construct
an interval total $(n+4)$-coloring of the graph $W_{n}$ for $n\geq
9$. We consider two cases.

Case 1: $n$ is even.

Define a total coloring $\alpha$ of the graph $W_{n}$ as follows:

1) $\alpha(u)=7$, $\alpha(v_{1})=1$, $\alpha(v_{2})=6$,
$\alpha(v_{3})=8$ and for $i=4,\ldots,\frac{n}{2}-2$, let $\alpha
\left(v_{i}\right) =2i+1$;

2) $\alpha(v_{\frac{n}{2}-1})=n+2$, $\alpha(v_{\frac{n}{2}})=n+4$
and for $j=\frac{n}{2}+1,\ldots,n-2$, let $\alpha(v_{j}) =2(n-j)$,
$\alpha(v_{n-1})=3$;

3) $\alpha(uv_{1})=3$, $\alpha(uv_{2})=5$ and for
$k=3,\ldots,\frac{n}{2}-1$, let
\begin{center}
$\alpha\left(uv_{k}\right)=2k+3$;
\end{center}

4) for $l=\frac{n}{2},\ldots,n-1$, let
\begin{center}
$\alpha\left(uv_{l}\right) =2(n-l+1)$;
\end{center}

5) $\alpha(v_{1}v_{2})=4$, $\alpha(v_{2}v_{3})=7$ and for
$p=3,\ldots,\frac{n}{2}-2$, let
\begin{center}
$\alpha\left(v_{p}v_{p+1}\right)=2(p+2)$;
\end{center}

6) for $q=\frac{n}{2}-1,\ldots,n-2$, let
\begin{center}
$\alpha\left(v_{q}v_{q+1}\right)=2(n-q)+1$ and $\alpha\left(
v_{1}v_{n-1}\right)=2$.
\end{center}

Case 2: $n$ is odd.

Define a total coloring $\beta$ of the graph $W_{n}$ as follows:

1) $\beta(u)=7$, $\beta(v_{1})=1$, $\beta(v_{2})=6$, $\beta
(v_{3})=8$ and for $i=4,\ldots,\lfloor\frac{n}{2}\rfloor-1$, let
$\beta\left(v_{i}\right) =2i+1$;

2) $\beta(v_{\lfloor\frac{n}{2}\rfloor})=n+4$,
$\beta(v_{\lceil\frac{n}{2}\rceil})=n+2$ and for
$j=\lceil\frac{n}{2}\rceil+1,\ldots,n-2$, let $\beta(v_{j})
=2(n-j)$, $\beta(v_{n-1})=3$;

3) $\beta(uv_{1})=3$, $\beta(uv_{2})=5$ and for
$k=3,\ldots,\lfloor\frac{n}{2}\rfloor$, let
\begin{center}
$\beta\left(uv_{k}\right) =2k+3$;
\end{center}

4) for $l=\lceil\frac{n}{2}\rceil,\ldots,n-1$, let
\begin{center}
$\beta \left(uv_{l}\right)=2(n-l+1)$;
\end{center}

5) $\beta(v_{1}v_{2})=4$, $\beta(v_{2}v_{3})=7$ and for
$p=3,\ldots,\lfloor\frac{n}{2}\rfloor$, let
\begin{center}
$\beta\left(v_{p}v_{p+1}\right)=2(p+2)$;
\end{center}

6) for $q=\lceil\frac{n}{2}\rceil,\ldots,n-2$, let
\begin{center}
$\beta\left(v_{q}v_{q+1}\right)=2(n-q)+1$ and
$\beta\left(v_{1}v_{n-1}\right)=2$.
\end{center}

It is easy to check that $\alpha$ is an interval total
$(n+4)$-coloring of the graph $W_{n}$, when $n$ is even, and $\beta
$ is an interval total $(n+4)$-coloring of the graph $W_{n}$, when
$n$ is odd.
 $~\square$
\end{proof}

\begin{figure}[h]
\begin{center}
\includegraphics[height=25pc,width=29pc]{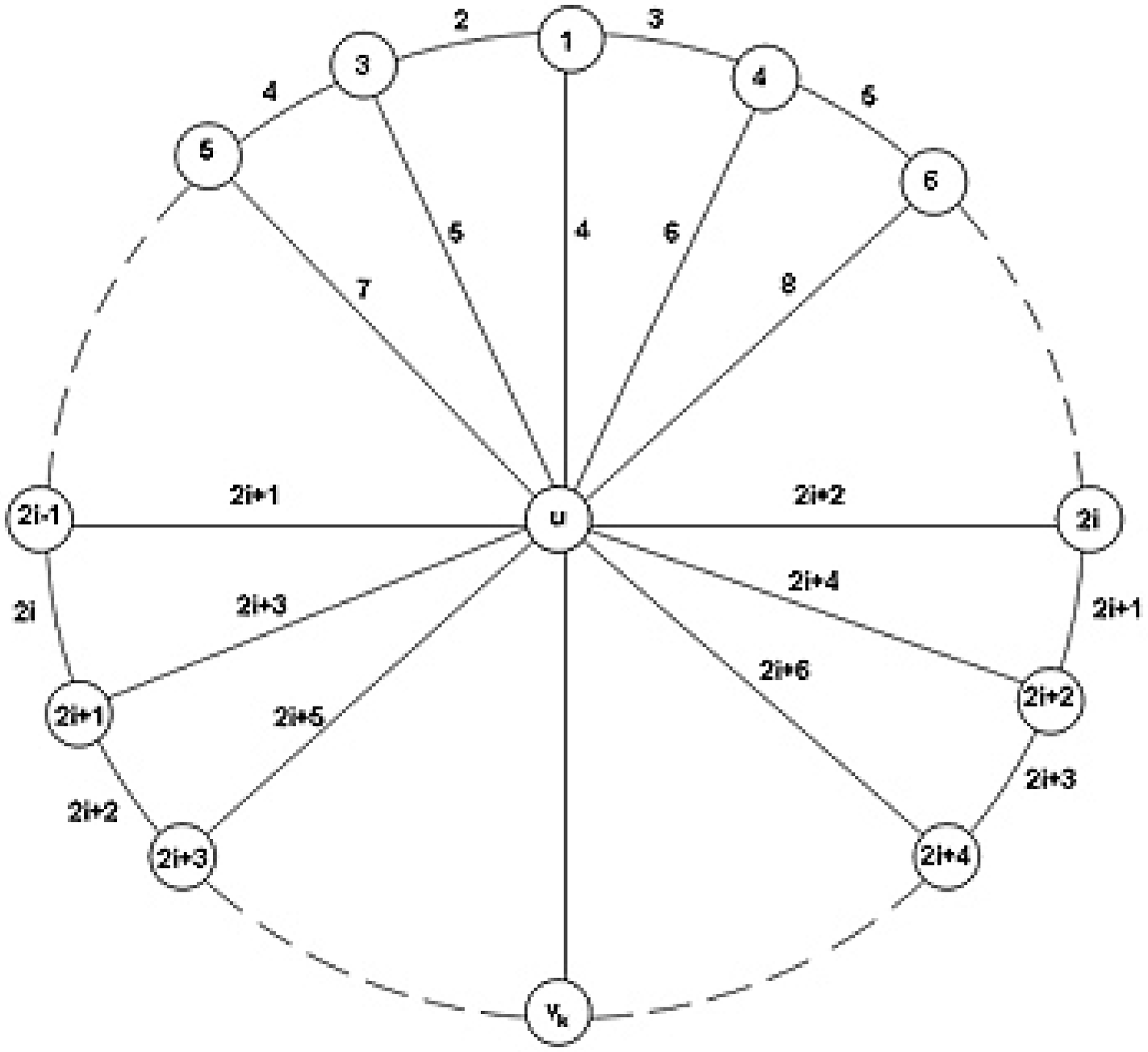}\
\caption{}
\end{center}
\end{figure}

\begin{lemma}
\label{mylemma5} For any $n\geq 4$, we have $W_{\tau}(W_{n})\leq
n+4$.
\end{lemma}
\begin{proof}
First, by Theorem \ref{mytheorem5}, we have $W_{\tau}(W_{n})\leq
n+6$ for any $n\geq 4$.

Next we prove that $W_{n}\notin \mathfrak{T}_{n+5}$.

Suppose, to the contrary, that $\alpha$ is an interval total
$(n+5)$-coloring of the graph $W_{n}$ for $n\geq 4$.

Consider the vertex $u$. Clearly,

\begin{center}
$1\leq \min S[u,\alpha]\leq 6$,
\end{center}
hence
\begin{center}
$n\leq \max S[u,\alpha]\leq n+5$.
\end{center}

Proposition \ref{myproposition1} implies that the following three
cases are possible:

1) $S[u,\alpha]=[6,n+5]$;

2) $S[u,\alpha]=[5,n+4]$;

3) $S[u,\alpha]=[4,n+3]$.

Case 1: $S[u,\alpha]=[6,n+5]$ or $S[u,\alpha]=[5,n+4]$.

Clearly, $\alpha(uv_{i})\geq 5$ for $i=1,\ldots,n-1$. This implies
that $\min S[v_{i},\alpha]\geq 2$ for $i=1,\ldots,n-1$, which is a
contradiction.

Case 2: $S[u,\alpha]=[4,n+3]$.

First we show that $\alpha(u)\neq 4$. Suppose that $\alpha(u)=4$.
This implies that $\alpha(uv_{i})\geq 5$ for $i=1,\ldots,n-1$, which
is a contradiction.

Let $e=uv_{1}$ and $\alpha(e)=4$. Note that $\alpha(v_{1})=1$.

Without loss of generality, we may assume that
$\alpha(v_{1}v_{2})=2$, $\alpha(v_{1}v_{n-1})=3$,
$\alpha(uv_{2})=5$, $\alpha(uv_{n-1})=6$, and there is a vertex
$v_{k}$ such that either $\alpha(v_{k})=n+5$ or
$\alpha(v_{k}v_{k+1})=n+5$ (see Fig. 1).

Let us consider simple paths

\begin{center}
$P_{1}=\left(v_{1},v_{1}v_{2},v_{2},\ldots,v_{k},v_{k}v_{k+1},v_{k+1}\right)$
and
$P_{2}=\left(v_{n-1},v_{n-1}v_{n-2},v_{n-2},\ldots,v_{k+1},v_{k+1}v_{k},v_{k}\right)$,
\end{center}
where $1\leq k\leq n-2$.

Let us show that

1) $\alpha(v_{i})=2i-1$, $\alpha(v_{i}v_{i+1})=2i$, $\alpha
(uv_{i})=2i+1$,

2) $\alpha(v_{n+1-i})=2i$, $\alpha(v_{n-i}v_{n+1-i})=2i+1$, $\alpha
(uv_{n+1-i})=2(i+1)$,\\
for $i=2,\ldots,k$.

We use induction on $i$. For $i=2$, it suffices to prove that
$\alpha(v_{2})=3$, $\alpha(v_{2}v_{3})=4$, $\alpha(v_{n-1})=4$,
$\alpha(v_{n-2}v_{n-1})=5$.

Consider the vertex $v_{2}$. Since $\alpha(v_{1}v_{2})=2$ and
$\alpha(uv_{2})=5$, we have $\min S[v_{2},\alpha]=2$ and $\max
S[v_{2},\alpha]=5$, hence $\{3,4\}\subseteq S[v_{2},\alpha]$. If we
suppose that $\alpha(v_{2})=4$, then $\alpha(v_{2}v_{3})=3$ and
$\max S[v_{3},\alpha]<7$, which contradicts $\max
S[v_{3},\alpha]\geq 7$. From this we have $\alpha(uv_{3})=7$ (see
Fig. 1).

Now we consider the vertex $v_{n-1}$. Since $\alpha(v_{1}v_{n-1})=3$
and $\alpha(uv_{n-1})=6$, we have $\min S[v_{n-1},\alpha]=3$ and
$\max S[v_{n-1},\alpha]=6$, hence $\{4,5\}\subseteq
S[v_{n-1},\alpha]$. If we suppose that $\alpha (v_{n-1})=5$, then
$\alpha(v_{n-2}v_{n-1})=4$ and $\max S[v_{n-2},\alpha]<8$, which
contradicts $\max S[v_{n-2},\alpha]\geq 8$ (see Fig. 1).

Suppose that the statements 1) and 2) are true for all $i^{\prime}$,
$1\leq i^{\prime}\leq i$. We prove that the statements 1) and 2) are
true for the case $i+1$, that is, $\alpha(v_{i+1})=2i+1$, $\alpha
(v_{i+1}v_{i+2})=2i+2$, $\alpha(uv_{i+1})=2i+3$ and $\alpha
(v_{n-i})=2i+2$, $\alpha(v_{n-i-1}v_{n-i})=2i+3$, $\alpha
(uv_{n-i})=2i+4$. From the induction hypothesis we have:

$1^{\prime}$) $\alpha(v_{j})=2j-1$, $\alpha(v_{j}v_{j+1})=2j$,
$\alpha(uv_{j})=2j+1$,

$2^{\prime}$) $\alpha(v_{n+1-j})=2j$, $\alpha
(v_{n-j}v_{n+1-j})=2j+1$, $\alpha
(uv_{n+1-j})=2(j+1)$,\\
for $j=2,\ldots,i$.

$1^{\prime}$) and $2^{\prime}$) implies that $\alpha(uv_{i+1})=2i+3$
and $\alpha(uv_{n-i})=2i+4$.

Consider the vertex $v_{i+1}$. Since $\alpha(v_{i}v_{i+1})=2i$ and
$\alpha(uv_{i+1})=2i+3$, we have $\min S[v_{i+1},\alpha]=2i$ and
$\max S[v_{i+1},\alpha]=2i+3$, hence $\{2i+1,2i+2\}\subseteq
S[v_{i+1},\alpha]$. If we suppose that $\alpha(v_{i+1})=2i+2$, then
$\alpha(v_{i+1}v_{i+2})=2i+1$ and $\max S[v_{i+2},\alpha]<2i+5$,
which contradicts $\max S[v_{i+2},\alpha]\geq 2i+5$. From this we
have $\alpha(uv_{i+2})=2i+5$ (see Fig. 1).

Next we consider the vertex $v_{n-i}$. Since
$\alpha(v_{n+1-i}v_{n-i})=2i+1$ and $\alpha(uv_{n-i})=2i+4$, we have
$\min S[v_{n-i},\alpha]=2i+1$ and $\max S[v_{n-i},\alpha]=2i+4$,
hence $\{2i+2,2i+3\}\subseteq S[v_{n-i},\alpha]$. If we suppose that
$\alpha(v_{n-i})=2i+3$, then $\alpha (v_{n-i-1}v_{n-i})=2i+2$ and
$\max S[v_{n-i-1},\alpha]<2i+6$, which contradicts $\max
S[v_{n-i-1},\alpha]\geq 2i+6$ (see Fig. 1).

By $1^{\prime}$), we have $k\geq \frac{n}{2}+2$.

By $2^{\prime}$), we have $k\leq \frac{n}{2}-1$.

It is easy to see that does not exist such an index $k$, which
satisfy the aforementioned inequalities. This completes the prove of
the case 2.

Similarly, it can be shown that $W_{n}\notin \mathfrak{T}_{n+6}$,
hence $W_{\tau}(W_{n})\leq n+4$ for any $n\geq 4$.
 $~\square$
\end{proof}

From Lemmas \ref{mylemma1}-\ref{mylemma5} and Remark
\ref{myremark2}, we have the following result:

\begin{theorem}
\label{mytheorem12} For $n\geq 4$, we have
\begin{description}
\item[(1)] $W_{n}\in \mathfrak{T}$,

\item[(2)] $w_{\tau}(W_{n})=\left\{
\begin{tabular}{ll}
$n+2$, & if $n=4$, \\
$n$, & if $n\geq 5$, \\
\end{tabular}%
\right.$

\item[(3)] $W_{\tau}(W_{n}) =\left\{
\begin{tabular}{ll}
$n+3$, & if $4\leq n \leq 8$,\\
$n+4$, & if $n\geq 9$,\\
\end{tabular}%
\right.$

\item[(4)] if $w_{\tau}(W_{n})\leq t\leq W_{\tau}(W_{n})$, then $
W_{n}\in \mathfrak{T}_{t}$.
\end{description}
\end{theorem}
\bigskip

\begin{acknowledgement}
We would like to thank Rafayel R. Kamalian for his attention to this
work.
\end{acknowledgement}

\end{document}